\documentclass[letterpaper,12pt]{article}
\usepackage[left=1.25in, right=1.25in, top=1.25in, bottom = 1.25in]{geometry}

\usepackage{amsmath,amsthm,amssymb,amscd}

\usepackage{dsfont}
\usepackage{mathrsfs}
\usepackage{setspace}
\usepackage{hyperref}

\usepackage{enumerate}
\usepackage{mathtools}
\usepackage{latexsym,amsfonts,bbm, mathabx}
\usepackage{color}
\usepackage{changes}

\newcommand{\DensityMatrixNbeta}{\varrho_N^{\beta,\omega}}
\newcommand{\DensityMatrixNbetaeins}{\varrho_N^{\beta,(1),\omega}}
\newcommand{\AverageParticleDensityNbeta}{\rho_N^{\beta,\omega}}
\newcommand{\AverageParticleDensityNjbeta}{\rho_{N_j}^{\beta,\omega}}
\newcommand{\EnsembleExpectation}[1]{\langle #1 \rangle_{\DensityMatrixNbeta}}

\newcommand{\ud}{\mathrm{d}}

\DeclareMathOperator{\tr}{Tr}

\DeclareMathOperator{\supp}{supp}

\numberwithin{equation}{section}

\newtheorem{theorem}{Theorem}[section]
\newtheorem{lemma}[theorem]{Lemma}

\newtheorem{remark}[theorem]{Remark}

\theoremstyle{definition}
\newtheorem{defn}[theorem]{Definition}

\newcommand{\laplace}{- \mathop{}\!\mathbin\bigtriangleup}

\newcommand{\e}{\mathrm{e}}

\numberwithin{equation}{section}

\begin{document}
	
	\allowdisplaybreaks[1]
	
	\thispagestyle{empty}
	
	\vspace*{1cm}
	
	\begin{center}
		
		{\Large \bf Bose--Einstein condensation for particles with repulsive short-range pair interactions in a Poisson random external potential in $\mathds R^d$} \\

		\vspace*{2cm}
		
		{\large  Joachim~Kerner \footnote{E-mail address: {\tt joachim.kerner@fernuni-hagen.de}} }%
		
		\vspace*{5mm}
		
		Department of Mathematics and Computer Science\\
		FernUniversität in Hagen\\
		58084 Hagen\\
		Germany\\
		
		\vspace*{2cm}

		{\large  Maximilian Pechmann \footnote{E-mail address: {\tt mpechmann@utk.edu}} }%
		
		\vspace*{5mm}
		
		Department of Mathematics\\
		University of Tennessee\\
		Knoxville, TN 37996\\
		USA\\

	\end{center}
	
	\vfill
	
	\begin{abstract} We study Bose gases in $d$ dimensions, $d \ge 2$, with short-range repulsive pair interactions, at positive temperature, in the canonical ensemble and in the thermodynamic limit. We assume the presence of hard Poissonian obstacles and focus on the non-percolation regime. For sufficiently strong interparticle interactions, we show that almost surely there cannot be Bose--Einstein condensation into a sufficiently localized, normalized one-particle state. The results apply to the eigenstates of the underlying one-particle Hamiltonian.
	\end{abstract}
	
	\newpage
	
	\section{Introduction}
	
	An important phenomenon in many-body quantum theory is Bose--Einstein condensation (BEC). It refers to a surprising coherent behavior in (possibly interacting) bosonic many-particle systems which occurs below some critical temperature, or equivalently, above some critical particle density. Originally predicted by Einstein to occur in non-interacting Bose gases in three dimensions \cite{EinsteinBECI,EinsteinBECII}, a rigorous proof of BEC for realistic interacting continuum system was achieved only some twenty years ago \cite{LiebSeiringerProof,LVZ03}. Since then, BEC remained a highly active area in mathematical physics. We refer to \cite{SchleinBECII, SchleinBECI,SchleinBECIII} and references therein for further information on current developments in BEC in a non-random setting. 
	
	An important open question regarding BEC is whether it is stable with respect to repulsive short-range interparticle interactions in the classical thermodynamic limit. Recently, we studied this question in a one-dimensional setting, namely, in the so-called Luttinger--Sy model where the external potential is a random (singular) potential generated by a Poisson point process on $\mathds R$ \cite{KernerPechmannHC}. In the present paper, it is our aim to generalize some of the results obtained there to the higher-dimensional setting. More explicitly, we study BEC in $2 \le d \in \mathds N$ dimensions in the canonical ensemble at positive temperature and in the presence of hard Poissonian obstacles, that is, hard balls of a fixed radius that are distributed according to a Poisson point process on $\mathds R^d$. We assume the intensity of the Poisson point process to be large enough such that no percolation is present. Regarding the interparticle interaction, we explore the hard core case, that is, each boson is assumed to be a hard ball, with a radius that can be constant or converge to zero at any speed. We will also consider more general repulsive pair interactions, modeled by a non-negative function with certain properties. In either case, whenever the pair interaction is sufficiently strong, we show that almost surely there cannot be BEC into a sufficiently localized one-particle state. As a consequence, almost surely there cannot be BEC into any one-particle eigenstate of the underlying one-particle Hamiltonian.
	
    We want to stress that hardcore interactions are not only interesting from a mathematical point of view since particles in realistic gases repel each other strongly at very short distances, as famously expressed by potentials of the Lennard--Jones type. In addition, at positive temperature and whenever the particle density is sufficiently large, BEC is expected to occur (in the grand-canonical ensemble) in a non-interacting Bose gas placed in a Poisson random potential~\cite{kac1973bose, kac1974bose, LenobleZagrebnovLuttingerSy}. Our results then show that such a condensate would be destroyed by the presence of sufficiently repulsive pair interactions. For this reason, it might prove interesting to study generalized BEC in such a scenario (as done, for example, in \cite{KPS18} for the one-dimensional case at zero temperature). In addition, it would be interesting to understand if, for example via the method of enlargement of obstacles (as illustrated, for example, in~\cite{sznitman1998brownian}), some of our results obtained for the non-percolation regime and hard Poissonian obstacles can be carried over to the percolation regime and soft Poissonian obstacles.
	
    The paper is organized as follows: In Section~\ref{The model} we introduce our model and in Section~\ref{Probabilistic results} we discuss probabilistic properties of our system that are used afterwards. In Section~\ref{Results on BEC} we then present our results regarding BEC; we discuss the case of hardcore interactions in Subsection~\ref{subsection BEC hard core} and the case of more general interactions in Subsection~\ref{subsection BEC more general case}.
	
	\section{The model}\label{The model}
	
	We study interacting Bose gases in $\mathbb{R}^d$, $2 \le d \in \mathds N$, and in an external Poisson random potential $V(\omega,x)$. Denoting the underlying probability space by $(\Omega,\Sigma,\mathbb{P})$, on an informal level the external potential reads
	\begin{equation}\label{ExtPot}
	V(\omega,x) :=\sum_{j}u(\| x-x_j^{\omega}\|_{\mathds R^d}) , \quad x \in \mathds R^d\ , \ \omega \in \Omega,
	\end{equation}
	where $\{x_j^{\omega}\}_{j}$ is a set of random points generated by a Poisson point process on $\mathds R^d$ with intensity $\nu > 0$. Furthermore, we assume that the single-site potential $u:\mathbb{R}^d \rightarrow \overline{\mathbb{R}}$ is given by
	\begin{equation}
	u(x):=\begin{cases}
	0 \ , \quad & \text{if} \quad x > R \ , \\
	\infty \ , & \text{else}\ ,
	\end{cases}
	\end{equation}
	where $R > 0$ is a constant. This means that we place hard balls $\overline{B_R(x_j^{\omega})}$ with radius $R > 0$ at each random point $x_j^{\omega}$. Note that such a random potential appears in well-known models such as the Kac--Luttinger model in the area of BEC \cite{kac1973bose,kac1974bose} and the Boolean--Poisson model in stochastic geometry \cite{GouereI}. 
	
	Also, we will investigate BEC in the thermodynamic limit: In this limit, $N$ bosons are placed in the cube $\Lambda_N:=(-L_N/2,+L_N/2)^d \subset \mathbb{R}^d$ of sidelength $L_N >0 $ such that the particle density 
	\begin{equation}
	\rho:=\frac{N}{L_N^d}
	\end{equation}
	remains constant in the limit $N \rightarrow \infty$. The $N$-particle configuration space in the external random potential \eqref{ExtPot} is given by $( \Lambda_N^{\omega} )^N$ with
	\begin{equation}\label{ConfigI}
	\Lambda_N^{\omega}:= (-L_N/2,+L_N/2)^d\setminus \bigcup_{j}\overline{B_{R}(x_j^{\omega})} \ ,
	\end{equation}
    representing the one-particle configuration space.

    Hardcore pair interactions are then introduced by further reducing the configuration space $(\Lambda_N^{\omega} )^N$. For this, we define the set
\begin{equation}
\Lambda^{(\mathrm{HC}), \omega}_{N}:=\{x=(x_j) \in 	\left(\Lambda_N^{\omega}\right)^N: \|x_i -x_j\|_{\mathbb{R}^d} > a_N \ ,\ i,j=1,...,N,\ i \neq j  \}
\end{equation}
where $(a_N)_{N \in \mathbb{N}} \subset (0,\infty)$ denotes the sequence of radii describing the range of the pair interaction. On a rigorous level, the $N$-particle Hamiltonian with hardcore pair interactions is given by the self-adjoint $dN$-dimensional Dirichlet Laplacian defined on  $\mathrm{L}^2_s(\Lambda^{(\mathrm{HC}), \omega}_{N})$; here the index $s$ refers to the totally symmetric subspace of $\mathrm{L}^2(\Lambda^{(\mathrm{HC}), \omega}_{N})$. On an informal level, the $N$-particle Hamiltonian with hardcore pair interaction is given by 
\begin{equation}\label{NBodyHamiltonianhard}
H_N^{\omega}:=\sum_{i=1}^{N}\left(\laplace_i+V(\omega,x_i)\right)+\sum_{1 \le i < j \le N}w_N^{\textrm{hc}}(\|x_i-x_j\|_{\mathbb{R}^d})
\end{equation}
where 
\begin{equation}\label{InteractionPotential}
w_N^{\textrm{hc}}(x):=\begin{cases}
0\ , \quad \text{if}\quad x > a_N\ , \\
\infty \ , \quad \text{else}\ .
\end{cases}
\end{equation}
We study hardcore pair interactions in Subsection~\ref{subsection BEC hard core}.

We also consider a more general class of repulsive pair interactions in Subsection~\ref{subsection BEC more general case}: For all $N \in \mathds N$, with $w_N \in L^{\infty}(\mathds R) \cap L^{1}(\mathds R,x^{d-1}\ud x)$ non-negative, we introduce the $N$-particle Hamiltonian 
\begin{equation}\label{NBodyHamiltoniansoft}
H_N^{\omega}:=\sum_{i=1}^{N}\left(\laplace_i+V(\omega,x_i)\right)+\sum_{1 \le i < j \le N}w_N(\|x_i-x_j\|_{\mathbb{R}^d})
\end{equation}
on the Hilbert space $L_{\mathrm{s}}^2((\Lambda_N^{\omega})^N)$. Finally, we write 
\begin{equation}\label{OneParticleHamil}
h_0^{\omega}:= \laplace +V(\omega,x)\ .
\end{equation}
for the underlying  one-particle Hamiltonian on $L^2(\Lambda_N^{\omega})$.
\section{Probabilistic results} \label{Probabilistic results}
Firstly, we note that the volume of the vacancy set $\Lambda_N^{\omega}$, see~\eqref{ConfigI}, tends to be a constant fraction of $\Lambda_N$ in the limit $N \to \infty$. More precisely, for any $\varepsilon > 0$ we have $\lim_{N \to \infty} \mathds P(| |\Lambda_N^{\omega}|/|\Lambda_N| - \e^{-\nu \omega_d R^d} | < \varepsilon ) = 1$, where $\omega_d$ is the volume of the unit ball in $d$ dimensions, see, for example, \cite[p. 147]{sznitman1998brownian}. Consequently, for any $0 < c <  \e^{-\nu \omega_d R^d}$ there is $\mathds P$-almost surely a subsequence $(N_j)_{j \in \mathds N} \subseteq \mathds N$ such that $|\Lambda_{N_j}^{\omega}| > c |\Lambda_N|$ for all but finitely many $j \in \mathds N$. Also, $\Lambda_N^{\omega}$ is possibly divided into components (regions), but $\mathds P$-almost surely has only finitely many components for each $N \in \mathds N$ \cite[Proposition 4.1]{meester1996continuum}. We denote the component of $\Lambda_N^{\omega}$ with the largest volume by $\Lambda_{N,>}^{(1), \omega}$ and its volume by $|\Lambda_{N,>}^{(1), \omega}|$.

Next, we estimate the volume of the largest component of the vacancy set $\Lambda_N^{\omega}$. Note that $\mathds P$-almost surely a ball free of Poisson points with radius $(\frac{d}{\nu \omega_d})^{1/d} (\ln L_N)^{1/d} - c$ for an arbitrary constant $c > 0$ occurs within $\Lambda_N$ for all but finitely many $N \in \mathds N$, for dimensions $d \ge 2$, see, for example, \cite[Proof of Proposition 4.4.3]{sznitman1998brownian}.
Recall that $\omega_d$ is the volume of the unit ball in $d$ dimensions. On the other hand, we have the following result.

\begin{theorem}\label{TheoremVacancySet}
	Let $2 \le d \in \mathds N$ be given. For any radius $R > 0$ of the hard Poissonian obstacles there is a $\tilde \nu > 0$ such that for all intensities $\nu > \tilde \nu$ of the Poisson random potential the following holds: There is a $\widetilde C_1 > 0$ and a $\widetilde C_2> 0$ such that for the number $A_N^{\omega}$ of disjoint boxes $[s j_1-\frac{s}{2}, s j_1+\frac{s}{2}) \times [s j_2-\frac{s}{2}, s j_2+\frac{s}{2}) \times \ldots \times [s j_d-\frac{s}{2}, s j_d+\frac{s}{2})$ where $j = (j_1, j_2, \ldots, j_d) \in \mathds Z^d$ and $s := R / \sqrt{d}$ that intersect any one component of the vacancy set within $\Lambda_N$ we have
	\begin{equation}
	\lim\limits_{N \to \infty} \mathds P(A_N^{\omega} \le C \ln(L_N)) = 1
	\end{equation}
	as well as
	\begin{equation}
	A_N^{\omega} \le C' \ln(L_N)
	\end{equation}
	$\mathds P$-almost surely for all but finitely many $N \in \mathds N$, for all $C > \widetilde C_1$ and for all $C' > \widetilde C_2$, respectively.
\end{theorem}
\begin{proof}
	We partition $\mathds R^d$ into the boxes $[s j_1-\frac{s}{2}, s j_1+\frac{s}{2}) \times [s j_2-\frac{s}{2}, s j_2+\frac{s}{2}) \times \ldots \times [s j_d-\frac{s}{2}, s j_d+\frac{s}{2})$ where $j = (j_1, j_2, \ldots, j_d) \in \mathds Z^d$ and $s := R / \sqrt{d}$. The centers of the boxes are then given by the points $(s j_1, s j_2, \ldots, s j_d)$ and we shall call them \textit{vertices}. Vertices with an Euclidean distance of $s$ are called \textit{adjacent} and consequently we obtain a discrete graph $\mathcal G$. A sequence $(v_i)_{j=1}^J$, $J \in \{ 1,2,...,\infty\}$, of vertices in $\mathcal G$ such that $v_i$ and $v_{i+1}$ are adjacent for all $j \in \{1,2,\ldots,J-1\}$ is called a \textit{path} in $\mathcal G$. A path in $\mathcal G$ is called \textit{finite} if $J < \infty$ and \textit{infinite} whenever $J=\infty$. 
	
	It is important to note that if a Poisson point $x_j^{\omega}$ is contained in such a box, then the box is not contained in $\Lambda_N^{\omega}$ (informally, this is equivalent to say that the external potential is infinitely high across the box). We call a vertex \textit{vacant} if the corresponding box does not contain any Poisson point, and \textit{occupied} if the corresponding box contains at least one Poisson point. In the same way we call a path in $\mathcal G$ vacant if the path contains only vacant vertices. 
	
	Furthermore, we shall assume that the intensity of the Poisson point process $\nu > 0$ is larger than the critical intensity
	\begin{equation}
	\nu_c := \inf\{ \nu > 0 : \theta^0(\nu) = 0 \}
	\end{equation}
	where
	\begin{equation}
	\theta^{0}(\nu) = \mathds P(\exists \text{ an infinite, self-avoiding, vacant path starting at } 0) \ .
	\end{equation}
	Note that $0 < \nu_c < \infty$, due to a Peierls argument and since the graph $\mathcal G$ is of finite degree, see \cite[p. 349]{kesten2002some}, \cite{hammersley1957percolation}.
	
	Now, let $W^{\omega}(v)$, $v \in s\mathds Z^d = (s j_1, s j_2, \ldots, s j_d)$, $j = (j_1, j_2, \ldots, j_d) \in \mathds Z^d$, be the union of all vertices that can be reached by a vacant path on $\mathcal G$ from $v$, and let $\#W^{\omega}(v)$ denote the number of vertices in $W^{\omega}(v)$. Due to \cite[Theorem 2]{kesten2002some}, \cite{aizenman1987sharpness} and \cite{menshikov1986coincidence}, there are constants $0 < C_1,C_2 < \infty$ such that for any $n \in \mathds N$,
	\begin{equation}\label{probability largest region}
	\mathds P(\#W^{\omega}(0) \ge n) \le C_1 \mathrm{e}^{-C_2 n} \ .
	\end{equation}
	We choose a $C > C_2^{-1}$ and set $n = C \ln((L_N +2)^d)$. Using inequality~\eqref{probability largest region}, we conclude that for any $N \in \mathds N$ the probability that the number of boxes $[s j_1-\frac{s}{2}, s j_1+\frac{s}{2}) \times [s j_2-\frac{s}{2}, s j_2+\frac{s}{2}) \times \ldots \times [s j_d-\frac{s}{2}, s j_d+\frac{s}{2})$ intersecting the largest vacancy set $\Lambda_{N,>}^{(1), \omega}$ is equal to or larger than $n$ is bounded from above by
	\begin{align}
	\begin{split}
	\sum\limits_{v \in s\mathds Z^d \cap (-\lceil L_N/2 \rceil , + \lceil L_N/2 \rceil)^d} \mathds P(\#W^{\omega}(v) \ge n) & \le s^d (L_N + 2)^d \mathds P(\#W^{\omega}(0) \ge n) \\
	& \le C_1 s^d [(L_N + 2)^d]^{1 - C C_2} \ ,
	\end{split}
	\end{align}
	which converges to zero in the limit $N \to \infty$, since $C > C_2^{-1}$. In addition, if $CC_2 > 2$, then, using the Borel--Cantelli lemma, we conclude that for $\mathds P$-almost all $\omega \in \Omega$ there exists an $\widetilde N \in \mathds N$ such that for all $N \ge \widetilde N$ the number of these boxes intersecting any component of the vacancy set within $\Lambda_N$ is smaller than $C \ln((L_N +2)^d)$.
\end{proof}
\begin{remark} \label{Remark 1}
	This theorem implies the following: Suppose that the intensity of the Poisson random potential is sufficiently large. Then the probability that the volume of the largest component $\Lambda_{N,>}^{(1), \omega}$ is bounded by $C\ln(L_N)$, for a constant $C > 0$ sufficiently large, converges to one; that is, there is a $\widetilde C > 0$ such that for all $C > \widetilde C$,
	\begin{align}
	\lim\limits_{N \to \infty} \mathds P(|\Lambda_{N,>}^{(1), \omega}| <  C \ln(L_N)) = 1 \ .
	\end{align}
	In addition, there is a $\widetilde C > 0$ such that for all $C > \widetilde C$ and for $\mathds P$-almost all $\omega \in \Omega$ there is an $\widetilde N \in \mathds N$ such that for all $N \ge \widetilde N$, $|\Lambda_{N,>}^{(1), \omega}| <  C \ln(L_N)$.
\end{remark}

For the proof of Lemma~\ref{EnergyDensity} in Subsection~\ref{subsection BEC more general case}, we need the following lemma. It is a statement about the number of disjoint balls with a given constant radius within $\Lambda_N$ that are free of Poisson points.
\begin{lemma} \label{Lemma 3.3}
	Let $d \ge 2$ and $\nu > 0$ be given. Also, let $(c_N)_{N \in \mathds N}$ be a sequence that goes to infinity. Then, for $\mathds P$-almost all $\omega \in \Omega$, there exists an $\widetilde N \in \mathds N$ such that for all $N \ge \widetilde N$ the number $B_N^{\omega}$ of disjoint balls with diameter $1$ that are completely within $\Lambda_N$ and are free of Poisson points is at least $L_N^d/(c_N \ln(N))$, that is, 
	\begin{align} \label{Inequality number of free balls}
	B_N^{\omega} & \ge \dfrac{L_N^d}{c_N \ln(N)} \ .
	\end{align}
\end{lemma} 
\begin{proof}
	We shall put $(2\lfloor L_N/2 \rfloor - 1)^d$ many disjoint balls each with diameter $1$ in $\Lambda_N$. More specifically, the balls shall have the centers $j = (j_1, j_2, \ldots, j_d) \in \mathds Z^d$ with $j_i \in [- \lfloor L_N/2 \rfloor + 1/2, \lfloor L_N/2 \rfloor - 1/2]$. 
	
	Next, we derive an upper bound on the probability that less than $\lceil L_N^d/(c_N \ln(L_N)) \rceil $ of these balls are free of Poisson points: If $0 <c < 1$ denotes the probability that one given ball is free of Poisson points, then $c= \e^{- \nu \pi^{d/2} 2^{-d} (\Gamma(d/2+1))^{-1}}$. Furthermore, 
	\begin{align}
	\begin{split}
	& \sum\limits_{i=0}^{\lceil L_N^d/(c_N \ln(L_N)) \rceil -1} \binom{(2\lfloor L_N/2 \rfloor -1)^d}{i} c^i (1-c)^{(2\lfloor L_N/2 \rfloor - 1)^d-i} \\
	& \quad \le \dfrac{L_N^d}{c_N \ln(L_N)} L_N^{\tfrac{d L_N^d}{c_N \ln(L_N)}} (1-c)^{{(2\lfloor L_N/2 \rfloor - 1)^d}-\tfrac{L_N^d}{c_N \ln(L_N)}} \\
	& \quad \le \exp\left\{d \ln(L_N) + \frac{dL_N^d}{c_N} + \Big[ (2\lfloor L_N/2 \rfloor - 1)^d-\frac{L_N^d}{c_N \ln(L_N)} \Big] \ln(1-c) \right\} \\
	& \quad \le \e^{3^{-d} L_N^d  \ln(1-c)}
	\end{split}
	\end{align}
	for all but finitely many $N \in \mathds N$. Since
	\begin{align}
	\sum\limits_{N \in \mathds N} \e^{3^{-d} L_N^d  \ln(1-c)} = \sum\limits_{N \in \mathds N} ((1-c)^{3^{-d} \rho^{-1}})^N \le \dfrac{1}{1 - (1-c)^{3^{-d} \rho^{-1}}} < \infty \ ,
	\end{align}
the claim follows with the lemma of Borel--Cantelli.
\end{proof}

We would like to comment on the main difference between Lemma~\ref{Lemma 3.3} and the corresponding one-dimensional result \cite[Lemma~A.1]{KernerPechmannHC}.
In the one-dimensional case, the lengths of the intervals that are introduced by a Poisson point process on the real line are independent, identically distributed random variables with exponential distribution. This fact was used in the proof  of Lemma~A.1 in \cite{KernerPechmannHC}. In higher dimensions, however, we know less about the distribution of the volume of the components of the vacancy set. To offset this, we require that the denominator in \eqref{Inequality number of free balls} converges to infinity at a sufficient speed, and needed to use a different strategy here compared to the corresponding one-dimensional case.

\section{Results on Bose--Einstein condensation} \label{Results on BEC}
In this section, we are going to apply the probabilistic results derived in Section~\ref{Probabilistic results} in order to say something about BEC in a system of interacting bosons placed in a random environment. In fact, we will consider two kinds of pair interactions: hardcore interactions where each particle is modeled as a hard ball with a certain radius and more general soft interactions. Physically, hardcore interactions are described by the informal Hamiltonian~\eqref{NBodyHamiltonianhard} and soft interactions by the Hamiltonian~\eqref{NBodyHamiltoniansoft}. As mentioned before, our aim is to generalize the results from~\cite{KernerPechmannHC} to the higher-dimensional setting.

In the canonical ensemble, the $N$-particle state of the system (the density matrix) at inverse temperature $\beta=1/T \in (0,\infty)$ is given by
\begin{equation}
\DensityMatrixNbeta=\frac{\mathrm{e}^{-\beta H_{N}^{\omega}}}{\tr(\mathrm{e}^{-\beta H_{N}^{\omega}})}\ .
\end{equation}
Here $\tr(\cdot)$ refers to the trace of a (trace-class) operator on the $N$-particle Hilbert space $\mathrm{L}^2_s(\Lambda^{(\mathrm{HC}), \omega}_{N})$; also, let  $\DensityMatrixNbeta(\cdot,\cdot)$ denote the kernel of $\DensityMatrixNbeta$. In order to calculate the density of particles in a given one-particle state, one uses the reduced one-particle density matrix which acts as a trace-class operator on the underlying one-particle Hilbert space $\mathrm{L}^2(\Lambda_N)$, see \cite{M07}. The kernel of the corresponding reduced one-particle density matrix is then obtained as
\begin{equation}
\DensityMatrixNbetaeins(x,y)=N \int_{\Lambda_{N}^{\omega}} \ud z_1 ... \int_{\Lambda_{N}^{\omega}} \ud z_{N-1}\ \DensityMatrixNbeta(x,z_1,...,z_{N-1},y,z_1,...,z_{N-1})\ 
\end{equation}
with $x,y \in \Lambda_{N}^{\omega}$. Consequently, the average particle density in a one-particle state $\varphi \in \mathrm{L}^2(\Lambda_{N}^{\omega})$ can be calculated as
\begin{equation}
\AverageParticleDensityNbeta(\varphi) := \frac{1}{L^d_N}\int_{\Lambda_{N}^{\omega}} \int_{\Lambda_{N}^{\omega}} \overline{\varphi(x)}\DensityMatrixNbetaeins(x,y) \varphi(y)\, \ud y \, \ud x
\end{equation} 
see, for example, \cite{PuleAonghusa87,M07}. This leads to the following definition. 
\begin{defn}\label{DefMacI} Let $\omega \in \Omega$ and $\varphi_N^{\omega} \in \mathrm{L}^2(\Lambda_{N}^{\omega})$ be a normalized one-particle state, $N \in \mathds N$. We call $(\varphi_{N}^{\omega})_{N \in \mathds N}$ macroscopically occupied at inverse temperature $\beta \in (0,\infty)$ if
	\begin{equation}
	\limsup_{N \rightarrow \infty} \AverageParticleDensityNbeta(\varphi_N^{\omega}) > 0 \ .
	\end{equation}
	In this case, we say that BEC into $(\varphi_{N}^{\omega})_{N \in \mathds N}$ is present.
\end{defn}

\subsection{Hardcore interactions} \label{subsection BEC hard core}

We firstly consider the $N$-particle Hamiltonian with hardcore pair interaction,
\begin{equation}
H_N^{\omega}=\sum_{i=1}^{N}\left(\laplace_i+V(\omega,x_i)\right)+\sum_{1 \le i < j \le N}w_N^{\textrm{hc}}(\|x_i-x_j\|_{\mathbb{R}^d})
\end{equation}
where 
\begin{equation}
w_N^{\textrm{hc}}(x):=\begin{cases}
0\ , \quad \text{if}\quad x > a_N\ , \\
\infty \ , \quad \text{else}\ .
\end{cases}
\end{equation}
We decompose $\mathds{R}^d$ into the boxes
\begin{equation}
\Lambda^{(n)}_N:=\{x \in \mathds{R}^d: r_Nn_j \leq x_j < r_N(n_j+1)\ , \ j=1,...,d \}  \ ,
\end{equation}
where $n = (n_1, n_2, \ldots, n_d) \in \mathds{Z}^d$ and $N \in \mathds N$. If the sidelength of these boxes satisfy $r_N \leq a_N/\sqrt{d}$, then any box $\Lambda^{(n)}_N$ can by occupied by at most one particle. Consequently, for a normalized one-particle state $\varphi_N^{\omega} \in L^2(\Lambda_N^{\omega})$, $N \in \mathds N$, $\omega \in \Omega$, we have that $\mathds P$-almost surely, for all $\beta \in (0,\infty)$, and for all but finitely many $N \in \mathds N$,
	\begin{equation} \label{inequality PuleAonghusa87}
	\AverageParticleDensityNbeta(\varphi_N^{\omega}) \leq \frac{1}{L^d_N}\Big(\sum_{n \in \mathds{Z}^d} \big( \int_{\Lambda^{(n)}_N}|\varphi_N^{\omega}(x)|^2\ \ud x \big)^{1/2} \Big)^2 \ ,
	\end{equation}
see \cite[Lemma~2]{PuleAonghusa87}. Note that each $\varphi_N^{\omega}$ in \eqref{inequality PuleAonghusa87} is understood to be extended by zero to all of $\mathds R^d$. We now can give and prove the main statement of this subsection.

\begin{theorem}[Absence of BEC I] \label{MainResultBEC1} Let $\beta \in (0,\infty)$ be arbitrarily given. We assume that $R > 0$ and $\nu > 0$ are such that Theorem~\ref{TheoremVacancySet} holds. Suppose that the sequence of radii $(a_N)_{N \in \mathds{N}}$ is such that
	\begin{equation}
	\lim\limits_{N \to \infty} \dfrac{1}{N} \left( \dfrac{\ln(N)}{a_N^d} \right)^2 = 0 \ .
	\end{equation}
	Then, for $\mathds P$-almost $\omega \in \Omega$, we have the following result: If $(\varphi_N^{\omega})_{N \in \mathds N}$, $\varphi_N^{\omega} \in L^2(\Lambda_N^{\omega})$ for all $N \in \mathds N$, is a sequence of normalized one-particle states for which the number $A_N^{\omega}$ of components of $\Lambda_N^{\omega}$ intersecting $\supp(\varphi_N^{\omega})$ satisfies
	\begin{equation}
	\lim_{N \to \infty} \dfrac{1}{N} \left( \dfrac{A_N^{\omega}\ln(N)}{a_N^d}\right)^2 = 0 \ ,
	\end{equation}
	then $(\varphi_N^{\omega})_{N \in \mathds N}$ is not macroscopically occupied, that is, there cannot exist a subsequence $(N_j)_{j \in \mathds N} \subseteq \mathds N$ such that
	\begin{equation}
	\lim_{j \rightarrow \infty} \AverageParticleDensityNjbeta(\varphi_{N_j}^{\omega}) > 0 \ . 
	\end{equation}
\end{theorem}
\begin{proof} The proof is obtained from a suitable adaptation of the proof of \cite[Theorem~3.3]{KernerPechmannHC}: Using inequality~\eqref{inequality PuleAonghusa87} and Theorem~\ref{TheoremVacancySet}, we have for a constant $C >0$ $\mathds P$-almost surely,
	\begin{equation}
	 \begin{split}
	\lim_{N \to \infty} \AverageParticleDensityNbeta(\varphi_N^{\omega}) &\leq \lim_{N \to \infty} \frac{1}{L^d_N}\left(\sum_{n \in \mathds{Z}^d}\left(\int_{\Lambda^{(n)}_N}|\varphi_N^{\omega}(x)|^2\ \ud x\right)^{1/2}\right)^2 \\
	&\leq \lim_{N \to \infty} \frac{1}{L^d_N}\left(\sum_{n \in \mathds{Z}^d: \supp(\varphi_N^{\omega})\cap \Lambda^{(n)}_N \neq \emptyset}1\right)^2 \\
	& \leq C \lim_{N \to \infty} \frac{1}{L^d_N } \left(\dfrac{A_N^{\omega} \ln(N)}{a_N^d} \right)^2 \\
	& = 0 \ .
	\end{split}
	\end{equation}
\end{proof}

In particular, Theorem~\ref{MainResultBEC1} shows that $\mathds P$-almost surely all eigenstates of the underlying one-particle Hamiltonian~\eqref{OneParticleHamil} are not macroscopically occupied if the interaction strength is large enough. This follows from the fact that these eigenstates are each supported on exactly one component of $\Lambda_N^{\omega}$.

\subsection{More general interactions}\label{subsection BEC more general case}

Finally, we study the case when a more general pair interactions is present. Namely, we now consider the $N$-particle Hamiltonian~\eqref{NBodyHamiltoniansoft}, that is,
\begin{equation}
H_N^{\omega}:=\sum_{i=1}^{N}\left(\laplace_i+V(\omega,x_i)\right)+\sum_{1 \le i < j \le N}w_N(\|x_i-x_j\|_{\mathbb{R}^d}) 
\end{equation}
where $w_N \in L^{\infty}(\mathds R) \cap L^{1}(\mathds R,x^{d-1}\ud x)$ is a non-negative function. We furthermore assume that for every $N \in \mathbb{N}$, there exist two numbers $a_N  > 0$ and $b_N > 0$ such that 
\begin{equation} \label{definition soft interaction}
w_N(x) \ge b_N \qquad \text{for almost every} \qquad x \in [-a_N,+a_N]\ .
\end{equation}
The main result of this section is the following theorem and says that any normalized one-particle states that is sufficiently localized, such as any eigenstate of the corresponding one-particle Hamiltonian~\eqref{OneParticleHamil}, is $\mathds P$-almost surely not macroscopically occupied, given that the pair interactions are sufficiently strong.

\begin{theorem}[Absence of BEC II]\label{CondensateXXX} Let $\nu > 0$ and $R > 0$ be given such that Theorem~\ref{TheoremVacancySet} holds. Let $\varphi_{N}^{\omega} \in \mathrm{L}^2(\Lambda_{N}^{\omega})$ be a normalized one-particle state and $A_N^{\omega}$ the number of components of $\Lambda_N^{\omega}$ with non-empty intersection with $\supp(\varphi_N^{\omega})$, $N \in \mathds N$, $\omega \in \Omega$. Furthermore, suppose that
\begin{equation}
\lim\limits_{N \to \infty} \frac{b_N a_N^{3d} N}{(A_N^{\omega})^3 \ln^3(N)} = \infty \ ,
\end{equation}
\begin{equation}
\lim\limits_{N \to \infty} \frac{a_N^{3d} N}{(A_N^{\omega})^3 \ln^3(N)} = \infty \ ,
\end{equation}
and
\begin{equation}
\lim_{N \to \infty} \ln^2 (N)\|w_N(\|\cdot\|_{\mathds{R}^d})\|_{\mathrm{L}^1(\mathds R^{d})} =0 \ .
\end{equation}
Then, for all $\beta \in (0,\infty)$, $(\varphi_{N}^{\omega})_{N \in \mathds N}$ is $\mathds{P}$-almost surely not macroscopically occupied.
\end{theorem}

In the rest of this work, we prove Theorem~\ref{CondensateXXX}. We proceed similarly as in \cite{KernerPechmannHC}. For the convenience of the reader and to be able to point out the differences between the higher-dimensional case discussed here and the one-dimensional case discussed in~\cite{KernerPechmannHC}, we present the main steps of the proof.

In a first step, we show that the expected energy density with respect to the canonical ensemble is bounded in the thermodynamic limit. However, in contrast to the one-dimensional setting in~\cite{KernerPechmannHC}, we have to assume that the pair interaction is weak enough in a suitable sense. This is due to the fact that, unlike in \cite[Lemma~A.1]{KernerPechmannHC}, we require the denominator in~\eqref{Inequality number of free balls} to converge to infinity at a certain speed.
\begin{lemma}[Bound energy density]\label{EnergyDensity} Let $\beta \in (0,\infty)$ be arbitrarily given. Assume that $\lim_{N \to \infty} \ln^2 (N) \|w_N(\|\cdot\|_{\mathds{R}^d})\|_{L^1(\mathds R^d)} =0$. Then there exists a constant $C>0$ such that $\mathds{P}$-almost surely, 
	\begin{equation} \label{equation lemma bound energy density}
	\limsup_{N \rightarrow \infty} \frac{\EnsembleExpectation{H_N^{\omega}}}{L^d_N} < C \ .
	\end{equation}
\end{lemma}
\begin{proof}
Let a typical $\omega \in \Omega$ be given. As in \cite{KernerPechmannHC}, we prove~\eqref{equation lemma bound energy density} by showing that the right side of the inequality
\begin{equation} \label{inequaitly lemma bound energy density}
   \frac{\beta}{2} \dfrac{\EnsembleExpectation{H_N^{\omega}} }{L_N^d} \le  \dfrac{\ln(\tr(\mathrm{e}^{-(\beta/2) H_N^{\omega}}))}{L_N^d}  - \dfrac{\ln(\tr(\mathrm{e}^{-\beta H_N^{\omega}}))}{L_N^d} 
  \end{equation}
is bounded by a constant in the limit $N \to \infty$. We note that the inequality~\eqref{inequaitly lemma bound energy density} holds due to the facts that $\EnsembleExpectation{H_N^{\omega}} = - \frac{ \mathrm{d}}{\mathrm{d} \beta} \ln(\tr(\mathrm{e}^{-\beta H_N^{\omega}}))$ and that $\ln(\tr(\mathrm{e}^{-\beta H_N^{\omega}}))$ is a convex function in $\beta$.

Regarding the first term in~\eqref{inequaitly lemma bound energy density}, we approximate the eigenvalues of $H_N^{\omega}$ by the eigenvalues $(E_N^{j})_{j \in \mathds N_0}$ of the Dirichlet Laplacian $\laplace$ on $\Lambda_N^N$ and conclude that $\tr(\mathrm{e}^{-(\beta/2) H_N^{\omega}}) \le \sum_{j \in \mathds N_0} \mathrm{e}^{-(\beta/2) E_N^{j}}$.
Then we use the fact that there is a constant $C_1 = C_1(\beta) > 0$ such that $\lim_{N \to \infty} L_N^{-d} \ln(\sum_{j \in \mathds N_0} \mathrm{e}^{-(\beta/2) E_N^{j}}) = C_1$, see, for example, \cite[Theorem~3.5.8]{ruelle1999statistical}. Thus, there is a constant $C_1 > 0$ such that $\mathds P$-almost surely
\begin{align}
 \limsup_{N \rightarrow \infty} \dfrac{\ln(\tr(\mathrm{e}^{-(\beta/2) H_N^{\omega}}))}{L_N^d} \le C_1 \ .
\end{align}

Next, we show that the second term~\eqref{inequaitly lemma bound energy density} (including the minus sign) is bounded from above by a constant in the limit $N \to \infty$ as well. We use the inequality \cite[Lemma 14.1 and Remark 14.2]{LiebSeiringer}
\begin{align}
  - \ln(\tr(\mathrm{e}^{-\beta H_N^{\omega}})) & \le -\ln \left( \mathrm{e}^{-\beta \| \Psi_N^{\omega}\|_{\mathds R^d}^{-2} \langle \Psi_N^{\omega}, H_N^{\omega} \Psi_N^{\omega} \rangle} \right) = \beta \| \Psi_N^{\omega}\|_{\mathds R^d}^{-2} \langle \Psi_N^{\omega}, H_N^{\omega} \Psi_N^{\omega} \rangle \ ,
\end{align}
$N \in \mathds N$, for a state $\Psi_N^{\omega}$ in the domain of $H_N^{\omega}$.
We choose the $N$-particle state $\Psi_N^{\omega}$ to be a product state $\prod_{j=1}^{N} \psi_N^{\omega}(x_j)$. The one-particle state $\psi_N^{\omega}(x)$, $x \in \mathds R^d$, shall be constructed as follows: One considers a rotational symmetric function $f(\|x\|_{\mathds R^d})$ that is equal to one for $\|x\|_{\mathds R^d} \le 1/4$ and smoothly decreases to zero between $1/4\le \|x\|_{\mathds R^d} \le 1/2$. Then $\psi_N^{\omega}$ is the sum of all such functions placed within each of the disjoint balls with diameter $1$ that are within $\Lambda_N$ and free of Poisson points. As in Lemma~\ref{Lemma 3.3}, we denote the number of such disjoint balls by $B_N^{\omega}$. Then,
\begin{align}
 \begin{split}
 & \|\Psi_N^{\omega}\|_{\mathds R^d}^{-2} \langle \Psi_N^{\omega}, H_N^{\omega} \Psi_N^{\omega} \rangle = N \|\psi_N^{\omega}\|_{L^2(\mathds R^d)}^{-2} \int\limits_{\Lambda_N} |(\psi_N^{\omega})'(x)|^2 \, \mathrm{d} x \\
 & \qquad \qquad + \, \binom{N}{2} \|\psi_N^{\omega}\|_{L^2(\mathds R^d)}^{-4} \int\limits_{\Lambda_N} \int\limits_{\Lambda_N} w_N(\|x-y\|_{\mathds R^d}) |\psi_N^{\omega}(x)|^2 |\psi_N^{\omega}(y)|^2 \, \mathrm{d} y \, \mathrm{d} x \ .
 \end{split}
\end{align}
Now, by construction there are positive constants $c_1,c_2 > 0$ independent of $N$ such that $\int_{\Lambda_N} |(\psi_N^{\omega})'(x)|^2 \, \mathrm{d} x \le c_1 B_N^{\omega}$, $\|\psi_N^{\omega}\|_{L^2(\mathds R^d)}^{2} \ge c_2 B_N^{\omega}$ as well as $|\psi_N^{\omega}(x)| \le 1$ for all $x \in \mathds R^d$. Employing Lemma~\ref{Lemma 3.3} in combination with our condition on $\|w_N(\|\cdot\|_{\mathds{R}^d})\|_{L^1(\mathds R^d)}$, we finally conclude that there is a constant $C_2> 0$ such that $\mathds P$-almost surely,
\begin{equation}
 \limsup\limits_{N \to \infty} \dfrac{-\ln(\tr(\mathrm{e}^{-\beta H_N^{\omega}}))}{L_N^d} \le \beta C_2 \ .
\end{equation}
\end{proof}
We now continue with the proof of Theorem~\ref{CondensateXXX}: Suppose there was a set $\widetilde \Omega \subset \Omega$ with $\mathds P(\widetilde \Omega) >0$ and such that for all $\omega \in \widetilde \Omega$ there was a sequence of normalized one-particle states $(\varphi_N^{\omega})_{N \in \mathds N} \in L^2(\Lambda_N^{\omega})$ that are macroscopically occupied and that fulfill the properties described in Theorem~\ref{CondensateXXX}. Then, for all such $\omega \in \widetilde \Omega$, we will show in the subsequent that the expected energy density with respect to the canonical ensemble would diverge in the thermodynamic limit. However, since this is in contradiction to Lemma~\ref{EnergyDensity},  Theorem~\ref{CondensateXXX} follows.

It remains to prove divergence of the energy density: Using second quantization, an equivalent definition for a sequence of one-particle states to be macroscopically occupied reads 
	\begin{equation}
	\limsup_{N \rightarrow \infty}\frac{\EnsembleExpectation{a^\ast(\varphi_{N}^{\omega})a(\varphi_{N}^{\omega})}}{L^d_N} > 0 \ .
	\end{equation}
Here, the creation operator $a^{\ast}(\varphi_N^{\omega}) = \int_{\Lambda_N^{\omega}} \varphi^{\omega}_N(x)a^{\ast}(x) \ud x$ and the annihilation operator $a(\varphi_N^{\omega}) = \int_{\Lambda_N^{\omega}} \varphi_N^{\omega}(x)a (x)\ud x$ fulfill the canonical commutation relations $\left[a(x),a^{\ast}(y)\right] = \delta(x-y)$ and $\left[a(x),a(y)\right] = \left[a^{\ast}(x),a^{\ast}(y)\right]=0$, $x,y \in \Lambda_N^{\omega}$. Also, $\delta(\cdot)$ is the Dirac-$\delta$ distribution.

The key ingredient now is a lower bound for the expected energy density. Let $\{G_j\}$ be a partition of $\mathds R^d$ into boxes of sidelength $a_N/\sqrt{d} > 0$ similar to the one as constructed in Theorem~\ref{TheoremVacancySet}.
By $K_N^{\omega}$ we denote the number of boxes in $\{G_j\}$ with a non-empty intersection with $\supp(\varphi_N^{\omega})$. Employing~\eqref{definition soft interaction}, we now obtain
	\begin{equation}
	\frac{\EnsembleExpectation{H_N^{\omega}}}{L^d_N} \geq \frac{b_N}{2L^d_N}\sum_{j}^{K_N^{\omega}} \int_{G_j}\ud x\int_{G_j}\ud y \ \EnsembleExpectation{a^{\ast}(x)a^{\ast}(y)a(x)a(y)}\ ,
	\end{equation}
	where the summation is over all the boxes in $\{G_j\}$ that have a non-empty intersection with $\supp(\varphi_N^{\omega})$. Introducing the functions $\varphi^{(j),\omega}_{N}:=\varphi_{N}^{\omega}\mathds 1_{G_j}$, $j=1,...,K_N^{\omega}$, as well as using \cite[(14)-(16b)]{PS86} and the inequality $\left|\sum_{i=1}^{n}x_j \right|^2 \leq n \sum_{i=1}^{n}|x_i|^2$, one gets
\begin{equation}
\frac{\EnsembleExpectation{H_N^{\omega}}}{L^d_N} \geq \frac{b_N}{2 L^d_N (K_N^{\omega})^3} \langle a^{\ast}(\varphi^{\omega}_{N})a(\varphi^{\omega}_{N}) \rangle^{2}_{\DensityMatrixNbeta}-\frac{b_N\rho}{2} \ .
\end{equation}

Finally, by Theorem~\ref{TheoremVacancySet}, $\mathds P$-almost surely and for all but finitely many $N \in \mathds N$ we have $K_N^{\omega} \le C A_N^{\omega} \ln(N)$ for some constant $C>0$. This finishes the proof.

\bibliographystyle{amsalpha}
\bibliography{Literature}

\def\cprime{$'$}\def\polhk\#1{\setbox0=\hbox{\#1}{{\o}oalign{\hidewidth
  \lower1.5ex\hbox{`}\hidewidth\crcr\unhbox0}}}\def\cprime{$'$}\def\polhk\#1{\setbox0=\hbox{\#1}{{\o}oalign{\hidewidth
  \lower1.5ex\hbox{`}\hidewidth\crcr\unhbox0}}}
\providecommand{\bysame}{\leavevmode\hbox to3em{\hrulefill}\thinspace}
\providecommand{\MR}{\relax\ifhmode\unskip\space\fi MR }
\providecommand{\MRhref}[2]{%
  \href{http://www.ams.org/mathscinet-getitem?mr=#1}{#2}
}
\providecommand{\href}[2]{#2}
\begin{thebibliography}{BBCS20}

\bibitem[AAS20]{SchleinBECIII}
C.~Brennecke A.~Adhikari and B.~Schlein, \emph{Bose–{E}instein condensation
  beyond the {G}ross--{P}itaevskii regime}, Ann.~Henri~Poincar\'{e} (2020).

\bibitem[AB87]{aizenman1987sharpness}
M.~Aizenman and D.~J. Barsky, \emph{Sharpness of the phase transition in
  percolation models}, Comm.~Math.~Phys \textbf{108} (1987), 489--526.

\bibitem[AP87]{PuleAonghusa87}
P.~M. Aonghusa and J.~V. Pul{\'e}, \emph{{Hard cores destroy {B}ose--{E}instein
  condensation}}, Lett.~Math.~Phys. \textbf{14} (1987), no.~2, 117--121.

\bibitem[BBCS18]{SchleinBECII}
C.~Boccato, C.~Brennecke, S.~Cenatiempo, and B.~Schlein, \emph{Complete
  {B}ose--{E}instein condensation in the {G}ross--{P}itaevskii regime}, Comm.
  Math. Phys. \textbf{359} (2018), no.~3, 975--1026. \MR{3784538}

\bibitem[BBCS20]{SchleinBECI}
\bysame, \emph{Optimal rate for {B}ose--{E}instein condensation in the
  {G}ross--{P}itaevskii regime}, Comm. Math. Phys. \textbf{376} (2020), no.~2,
  1311--1395. \MR{4103969}

\bibitem[dS86]{PS86}
P.~de~Smedt, \emph{{The Effect of Repulsive Interactions on {B}ose--{E}instein
  Condensation}}, J. Stat. Phys. \textbf{45} (1986), 201--213.

\bibitem[Ein24]{EinsteinBECI}
A.~Einstein, Sitzber. {K}gl. {P}reuss. {A}kadm. {W}iss. (1924), 261--267.

\bibitem[Ein25]{EinsteinBECII}
\bysame, Sitzber. {K}gl. {P}reuss. {A}kadm. {W}iss. (1925), 3--14.

\bibitem[Gou08]{GouereI}
J.~B. Gou\'{e}r\'{e}, \emph{Subcritical regimes in the {P}oisson {B}oolean
  model of continuum percolation}, Ann.~Prob. \textbf{36} (2008), no.~4.

\bibitem[Ham57]{hammersley1957percolation}
J.~M. Hammersley, \emph{Percolation processes: Lower bounds for the critical
  probability}, Ann. Math. Stat. \textbf{28} (1957), no.~3, 790--795.

\bibitem[Kes02]{kesten2002some}
H.~Kesten, \emph{Some highlights of percolation}, arXiv preprint math/0212398
  (2002).

\bibitem[KL73]{kac1973bose}
M.~Kac and J.~M. Luttinger, \emph{{Bose--Einstein condensation in the presence
  of impurities}}, J.~Math.~Phys. \textbf{14} (1973), no.~11, 1626--1628.

\bibitem[KL74]{kac1974bose}
\bysame, \emph{{Bose--Einstein condensation in the presence of impurities.
  II}}, J.~Math.~Phys. \textbf{15} (1974), no.~2, 183--186.

\bibitem[KP21]{KernerPechmannHC}
J.~Kerner and M.~Pechmann, \emph{On the effect of repulsive pair interactions
  on {B}ose--{E}instein condensation in the {L}uttinger--{S}y model}, Proc. Am.
  Math. Soc. (2021).

\bibitem[KPS19]{KPS18}
J.~Kerner, M.~Pechmann, and W.~Spitzer, \emph{{Bose--Einstein condensation in
  the Luttinger--Sy model with contact interaction}}, Ann. Henri Poincar{\'e}
  \textbf{20} (2019), 2101--2134.

\bibitem[LS02]{LiebSeiringerProof}
E.~H. Lieb and R.~Seiringer, \emph{{Proof of {B}ose--{E}instein condensation
  for dilute trapped gases}}, Phys. Rev. Lett. \textbf{88} (2002), 170409.

\bibitem[LS10]{LiebSeiringer}
\bysame, \emph{{The Stability of Matter in Quantum Mechanics}}, Cambridge
  University Press, 2010.

\bibitem[LVZ03]{LVZ03}
J.~Lauwers, A.~Verbeure, and V.~A. Zagrebnov, \emph{{Proof of
  {B}ose--{E}instein condensation for interacting gases with a one-particle
  gap}}, J.~Phys.~A \textbf{36} (2003), 169--174.

\bibitem[LZ07]{LenobleZagrebnovLuttingerSy}
O.~Lenoble and V.~A. Zagrebnov, \emph{{Bose--{E}instein condensation in the
  {L}uttinger--{S}y model}}, Mark.~Proc.~Rel.~Fields \textbf{13} (2007), no.~2,
  441--468.

\bibitem[Men86]{menshikov1986coincidence}
M.~V. Menshikov, \emph{Coincidence of critical points in percolation problems},
  Soviet Mathematics Doklady, vol.~33, 1986, pp.~856--859.

\bibitem[Mic07]{M07}
A.~Michelangeli, \emph{{Reduced density matrices and Bose--Einstein
  condensation}}, SISSA \textbf{39} (2007).

\bibitem[MR96]{meester1996continuum}
R.~Meester and R.~Roy, \emph{Continuum percolation}, Cambridge University
  Press, 1996.

\bibitem[Rue99]{ruelle1999statistical}
D.~Ruelle, \emph{{Statistical mechanics: Rigorous results}}, World Scientific,
  1999.

\bibitem[Szn98]{sznitman1998brownian}
A.-S. Sznitman, \emph{Brownian motion, obstacles and random media}, Springer,
  1998.

\end{thebibliography}
	
\end{document}